\documentclass[letterpaper,11pt]{article}

\usepackage{times}

\usepackage{amsthm,amsmath,amsfonts}

\usepackage{graphicx}
\usepackage[small]{caption}
\usepackage[hidelinks]{hyperref}
\newcommand{\orcid}[1]{\,\href{https://orcid.org/#1}{\includegraphics[width=8pt]{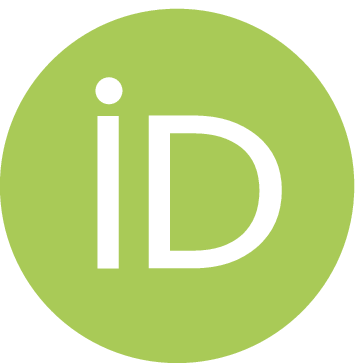}}}

\usepackage{color}

\renewcommand{\S}{\mathcal{S}}

\newcommand{\scirc}{\textsc{Simple Circuit}}
\newcommand{\spath}{\textsc{Simple Path}}

\newtheorem{theorem}{Theorem}
\newtheorem{lemma}{Lemma}
\theoremstyle{remark}

\title{Linking disjoint axis-parallel segments into a simple polygon\\is hard too}

\author{Rain Jiang\orcid{0000-0002-0144-942X}\qquad
Kai Jiang\orcid{0000-0001-8165-0571}\qquad
Minghui Jiang\orcid{0000-0003-1843-9292}\,\thanks{\texttt{ dr.minghui.jiang at gmail.com}}\medskip\\
Home School, USA}

\date{}

\begin{document}

\maketitle

\begin{abstract}
	Deciding whether a family of disjoint axis-parallel line segments in the plane
	can be linked into a simple polygon (or a simple polygonal chain)
	by adding segments between their endpoints is NP-hard.
\end{abstract}

\section{Introduction}

Given a family $\S$ of $n$ closed line segments in the plane,
\scirc\ (respectively, \spath)
is the problem of deciding whether
these segments can be linked into
a simple polygon (respectively, a simple polygonal chain)
by adding segments between their endpoints.

Rappaport~\cite{Ra89} proved that \scirc\ is NP-hard
if the segments in $\S$ are allowed to intersect at their common endpoints,
and asked whether the problem remains NP-hard when the segments are disjoint.
Later, Bose, Houle, and Toussaint~\cite{BHT01} asked
whether the related problem \spath\ is NP-hard
when the segments in $\S$ are disjoint.
Later, T\'oth~\cite{To06} asked about the complexity of \spath\ again,
with a special interest in the case
when the segments in $\S$ are both disjoint and axis-parallel.

Recently, Akitaya et~al.~\cite{AKRST19}
proved that \scirc\ is NP-hard
when the segments in $\S$ are disjoint and have only four distinct orientations.
Subsequently, the authors of this paper came up with a similar construction~\cite{JJJ21},
and proved that both \scirc\ and \spath\ are NP-hard,
when the segments in $\S$ are disjoint and have only four distinct orientations.
Akitaya et~al.~\cite{AKRST19} asked whether \scirc\ is NP-hard
when the segments in $\S$ are disjoint and axis-parallel.

In this paper,
we prove the following theorem:

\begin{theorem}\label{thm:hard}
	\scirc\ and \spath\ are both NP-hard
	even if the segments in $\S$ are disjoint and axis-parallel.
\end{theorem}

We prove the theorem in two steps.
First, we modify the construction in Rappaport's proof
of NP-hardness of \scirc\ on not necessarily disjoint axis-parallel segments~\cite{Ra89}
to show that the problem remains NP-hard on disjoint axis-parallel segments.
Next, we modify the construction further
to prove that \spath\ is also NP-hard on disjoint axis-parallel segments.

Some of the ideas behind our construction in this paper
are also used in our previous construction~\cite{JJJ21};
interested readers may read the short proof there as a warm-up exercise.
Following the same setup there,
we briefly review Rappaport's proof in the following,
which is based on a polynomial reduction from the NP-hard problem
\textsc{Hamiltonian Path} in planar cubic graphs.

\medskip
For any family $\S$ of closed segments in the plane,
denote by $V(\S)$ the set of endpoints of the segments in $\S$.
For any two endpoints $p$ and $q$ in $V(\S)$,
we call the open segment $pq$ a \emph{visibility edge}
if it does not intersect any closed segment in $\S$.

\medskip
Given a planar cubic graph $G$ with $n \ge 4$ vertices,
the reduction~\cite{Ra89} first obtains a rectilinear planar layout of the planar graph
using an algorithm of Rosenstiehl and Tarjan~\cite{RT86},
then constructs a family $\S$ of $O(n)$ axis-parallel segments
following the rectilinear planar layout,
such that $G$ admits a Hamiltonian path
if and only if $\S$ can be linked into a simple polygon
by adding visibility edges between endpoints in $V(\S)$.
The segments in $\S$ are axis-parallel and interior-disjoint,
but may intersect at common endpoints.
Each endpoint in $V(\S)$ is incident to at most one horizontal segment
and at most one vertical segment in $\S$.
Since the rectilinear planar layout
has width at most $2n-4$ and height at most $n$~\cite{RT86},
all coordinates of endpoints in $V(\S)$ are integers of magnitude
$O(n^2)$;
indeed a closer look at the construction~\cite[Figure~8]{Ra89}
shows that
$V(\S) \subseteq [1, 22n^2]\times [1, 11n]$.
The reduction is hence strongly polynomial.
Consequently, \scirc\ is strongly NP-hard,
on not necessarily disjoint axis-parallel segments.

\section{Modification for \scirc}

To show that \scirc\ remains strongly NP-hard on disjoint segments,
we will transform the family $\S$ of interior-disjoint axis-parallel segments,
which Rappaport constructed,
into a family $\S'$ of disjoint axis-parallel segments
in polynomial time,
such that $\S$ can be linked into a simple polygon
if and only if $\S'$ can be linked into a simple polygon.
This transformation can be viewed as a reduction from \scirc\ on one type
of input to the same problem on another type of input.

\begin{figure}[p]
	\centering\includegraphics{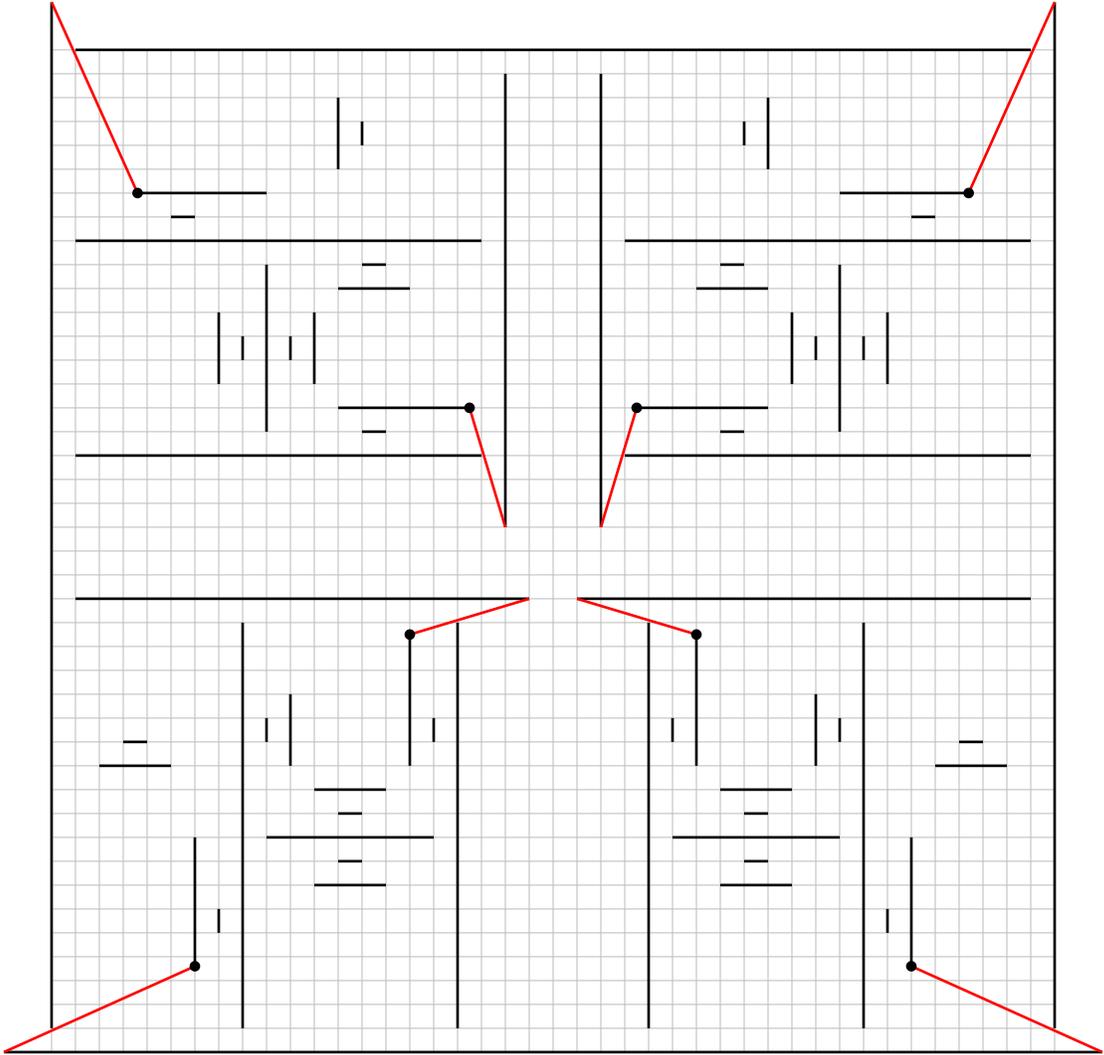}
	\caption{Gadgets of four different orientations for the four intersections
	(i.e., the four corners) of a unit square,
	scaled up by a factor of $42$ and illustrated on a $42\times 42$ grid.}\label{fig:all}
\end{figure}

To obtain $\S'$ from $\S$,
we first scale the integer coordinates of all segment endpoints
by a factor of $42$,
then locally modify each intersection between a horizontal segment
and a vertical segment into a gadget.
The gadgets come in four variants, one for each possible orientation of an intersection;
see Figure~\ref{fig:all}.
\clearpage

\begin{figure}[htb]
	\begin{center}
	\includegraphics{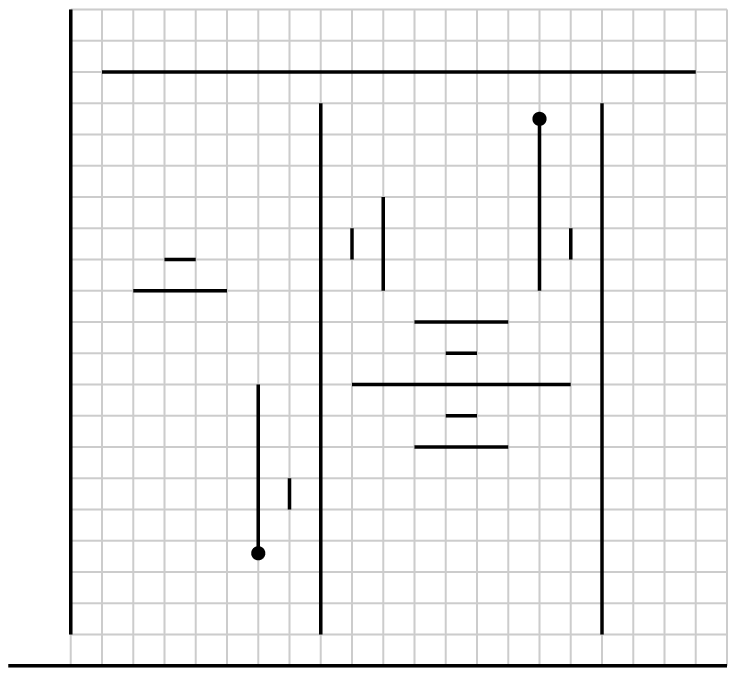}\includegraphics{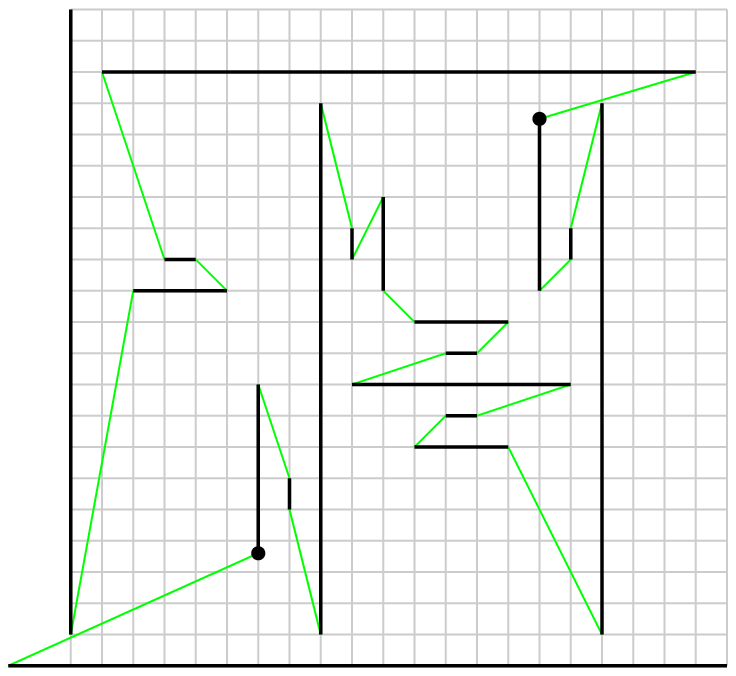}
	\end{center}
	\caption{The gadget for the intersection $o$ between a horizontal segment $oa$
	and a vertical segment $ob$ illustrated on a $21\times 21$ grid.
	Left: The common endpoint $o = (0, 0)$ is split to $o' = (-2,0)$, $o'' = (0, 1)$.
	The two endpoints illustrated as black dots are
	$u' = (6, 4 - \delta)$ and $v' = (15,17 + \frac13 + \delta)$.
	The $16$ added segments include
	the left group of five vertical segments,
	the middle group of five horizontal segments,
	the right group of three vertical segments,
	the top horizontal segment by itself,
	and the left group of two horizontal segments.
	The top horizontal segment has left endpoint $(1, 19)$ and right endpoint $(20, 19)$.
	Right: An alternating path of segments and visibility edges in the gadget.}\label{fig:gadget}
\end{figure}

By symmetry, it suffices to describe in detail only one variant of the gadget.
Refer to Figure~\ref{fig:gadget}.
For the intersection $o$ between a horizontal segment $oa$ and a vertical segment $ob$,
where $o$ is the left endpoint of $oa$ and the lower endpoint of $ob$,
the corresponding gadget is constructed as follows:
\begin{itemize}\setlength\itemsep{0pt}
		\item
			Separate the two segments $oa$ and $ob$ into two disjoint segments
			$o'a$ and $o''b$,
			by splitting their common endpoint $o = (0,0)$,
			then moving one to $o' = (-2,0)$ and the other to $o'' = (0,1)$.

		\item
			Add $16$ segments with integer coordinates as indicated by the grid lines in the figure,
			except that
		\begin{itemize}\setlength\itemsep{0pt}
				\item
					the lower endpoint of the vertical segment at $x = 6$ is $u' = (6, 4 - \delta)$,
				\item
					the upper endpoint of the vertical segment at $x = 15$ is
					$v' = (15,17 + \frac13 + \delta)$,
		\end{itemize}
		where $\delta$ is determined by Lemma~\ref{lem:angle} below.
\end{itemize}

In the presence of $\S'$,
we say that two points $p$ and $q$ can \emph{see} each other
if the open segment $pq$ is disjoint from all closed segments in $\S'$,
and we say that two segments $A$ and $B$ in $\S'$ can \emph{see} each other
if at least one of the four pairs of endpoints,
one of $A$ and one of $B$,
can see each other.
The gadgets we constructed have the following property of mutual invisibility:

\begin{lemma}\label{lem:angle}
	With $\delta = 1/(c\cdot n^2)$ for a sufficiently large integer constant $c > 0$,
	the endpoint $u' = (6,4-\delta)$ cannot see any endpoints
	outside the gadget through the gap between $o' = (-2,0)$ and $o'' = (0, 1)$,
	and
	the endpoint $v' = (15,17+\frac13+\delta)$ cannot see any endpoints
	outside the gadget through the gap between $(17,18)$ and $(20, 19)$.
\end{lemma}

\begin{figure}[htb]
	\centering\includegraphics{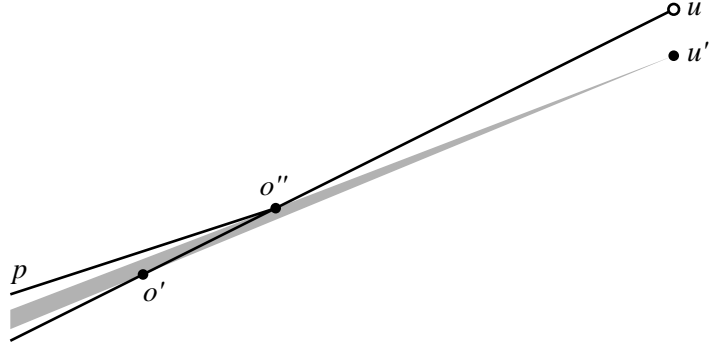}
	\caption{The narrow viewing angle from $u'$ through the gap between $o'$ and $o''$.}\label{fig:angle}
\end{figure}

\begin{proof}
	Refer to Figure~\ref{fig:angle}.
	Let $u = (6,4)$, which is collinear with $o' = (-2,0)$ and $o'' = (0,1)$.
	Note that
	$|u u'| = \delta$,
	$|o'' u'| > 6$,
	and $\angle o'' u' u > \pi/2$.
	Thus
	$$
	\angle u o'' u' < \tan\angle u o'' u' < \frac{|u u'|}{|o'' u'|} < \frac{\delta}6.
	$$

	Consider the ray that starts from $o''$ and goes through $o'$.
	Rotate this ray clockwise about $o''$ for a positive angle
	till it goes through another integer point $p$ in $V(\S')$.
	The area of the triangle $o'o''p$ is at least $\frac12$
	since all three endpoints have integer coordinates.
	Recall that the endpoints of all segments in $\S$ are in the range
	$[1, 22n^2]\times [1, 11n]$,
	and is scaled by a factor of $42$
	by the transformation from $\S$ to $\S'$.
	Thus we have
	$|po''| < 42\cdot (22n^2 + 11n) < 50 \cdot 33 n^2$.
	Also note that $|o'o''| < 3$.
	Thus
	$$
	\angle o'o''p > \sin\angle o'o''p
	= \frac{2\cdot\mathrm{area}(o'o''p)}{|o'o''|\cdot |po''|}
	> \frac1{3 \cdot 50 \cdot 33 n^2}
	> \frac1{5000 n^2}.
	$$

	Let $\delta = 1 / (c\cdot n^2)$ for a sufficiently large integer $c > 0$ such that
	$\angle u o'' u' < \angle o'o''p$.
	Then the ray $u' o''$ splits the angle $\angle o'o''p$.
	Thus the cone with the angle $\angle o' u' o''$ does not contain
	in its interior any integer endpoints in $V(\S')$.
	Thus $u'$ cannot see any integer endpoints through the gap between $o'$ and $o''$.

	Let $v = (15,17+\frac13)$, which is collinear with $(17,18)$ and $(20,19)$.
	Then $|v v'| = \delta$.
	By a similar analysis, we can also guarantee that
	$v'$ cannot see any integer endpoints from other gadgets through the gap
	between $(17,18)$ and $(20,19)$.

	Moreover,
	for an intersection of the orientation as illustrated in Figure~\ref{fig:gadget},
	the slope of a viewing line through $u'$ and its gap is around $\frac12$,
	and the slope of a viewing line through $v'$ and its gap is around $\frac13$.
	With $\delta = 1 / (c\cdot n^2)$ for a sufficiently large integer $c > 0$,
	no two endpoints $u'$ or $v'$ from two different gadgets can see each other,
	because the eight narrow ranges of slopes of viewing lines,
	as illustrated by the red lines in Figure~\ref{fig:all},
	are disjoint for intersections of different orientations.
\end{proof}

The following lemma shows that our local transformation preserves linkability:

\begin{lemma}\label{lem:link}
	$\S$ can be linked into a simple polygon
	if and only if $\S'$ can be linked into a simple polygon.
\end{lemma}

\begin{proof}
	We first prove the direct implication.
	Suppose that $\S$ can be linked into a simple polygon.
	For each visibility edge of this polygon between two endpoints in $V(\S)$,
	we add a visibility edge between the corresponding endpoints in $V(\S')$.
	For each endpoint $o$ incident to two segments $oa$ and $ob$ in $\S$,
	we link the segments in the corresponding gadget in $\S'$
	following the alternating path as illustrated in Figure~\ref{fig:gadget} right.
	Then $\S'$ is also linked into a simple polygon.

	We next prove the reverse implication.
	Suppose that $\S'$ can be linked into a simple polygon.
	Refer to Figures~\ref{fig:gadget1},
	\ref{fig:gadget2}, \ref{fig:gadget3},
	\ref{fig:gadget4},
	\ref{fig:gadget5},
	\ref{fig:gadget6}.
	The linking of the segments along the sequence of lengths as illustrated
	in Figures~\ref{fig:gadget1} through~\ref{fig:gadget6} is inevitable in each gadget.
	Then, following the other visibility edges of the simple polygon through $\S'$,
	which are outside and between the gadgets,
	$\S$ can be linked into a simple polygon too.
\end{proof}

\begin{figure}[p]
	\centering\includegraphics{gadget.eps}\includegraphics{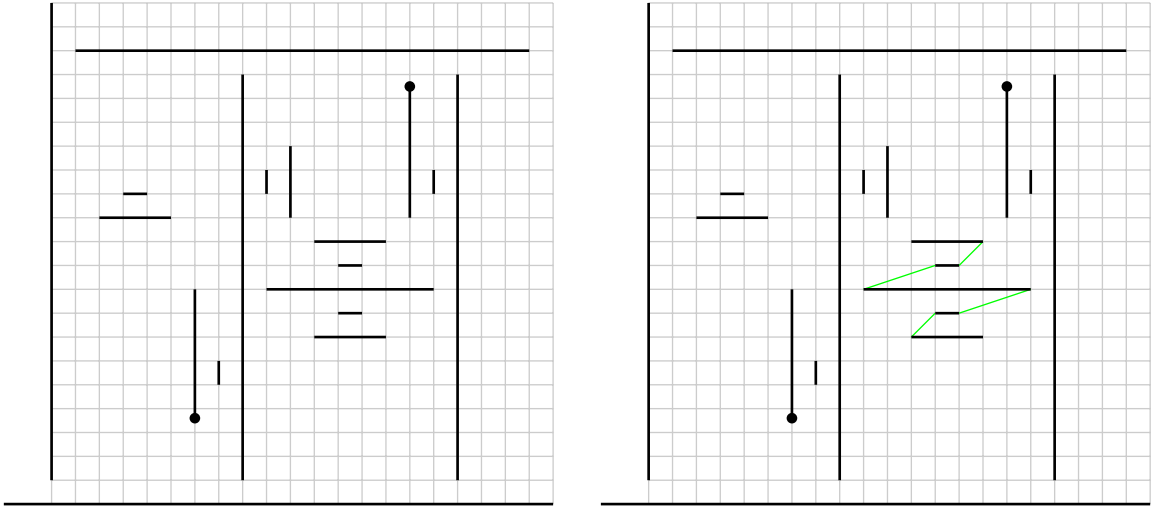}
	\caption{Linking the middle group of five horizontal segments.
	Each of the two length-$1$ segments
	can only see the length-$7$ segment between them and one other segment.
	Thus the five segments must be linked consecutively,
	forming a sequence of $3,1,7,1,3$ in lengths.
	There is some flexibility in the choice of which endpoints to link between two consecutive
	segments in the sequence.
	The combination of visibility edges illustrated here and in subsequent figures
	may be one of many possibilities unless specified.}\label{fig:gadget1}
\end{figure}

\begin{figure}[p]
	\centering\includegraphics{gadget1.eps}\includegraphics{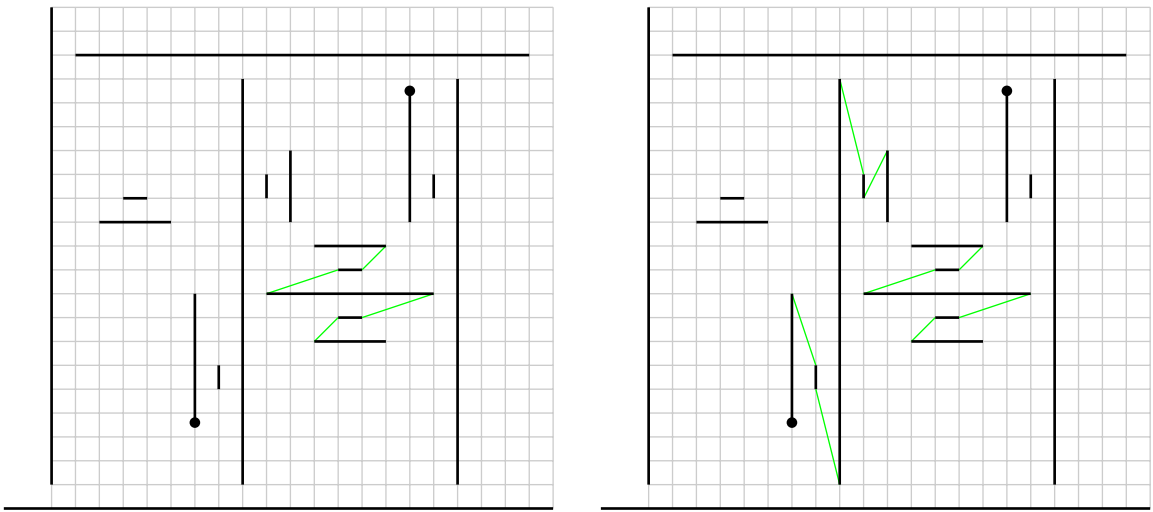}
	\caption{Linking the left group of five vertical segments.
	Among the unlinked neighbors,
	each of the two length-$1$ segments
	can only see the length-$17$ segment between them and one other segment.
	Thus the five segments must also be linked consecutively,
	forming a sequence of $5^+,1,17,1,3$ in lengths.}\label{fig:gadget2}
\end{figure}

\begin{figure}[p]
	\centering\includegraphics{gadget2.eps}\includegraphics{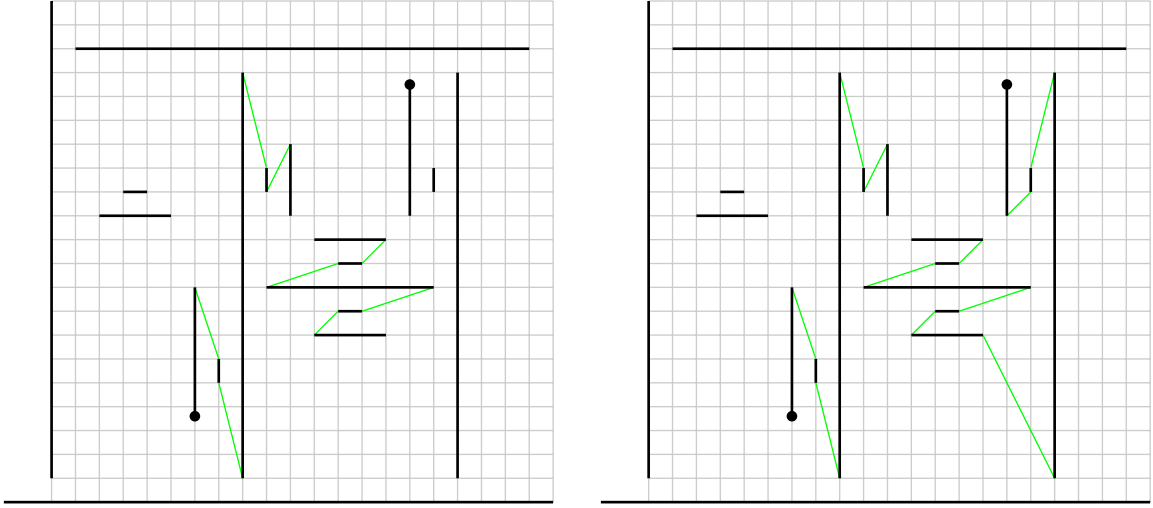}
	\caption{Linking the right group of three vertical segments.
	The lowest segment in the middle group of five horizontal segments must now be linked
	to the lower endpoint of the length-$17$ vertical segment.
	Consequently, the length-$1$ segment in the right group of three vertical segments
	can only be linked to the other two segments in the same group,
	and the sequence extends to $3,1,7,1,3,17,1,5^+$.}\label{fig:gadget3}
\end{figure}

\begin{figure}[p]
	\centering\includegraphics{gadget3.eps}\includegraphics{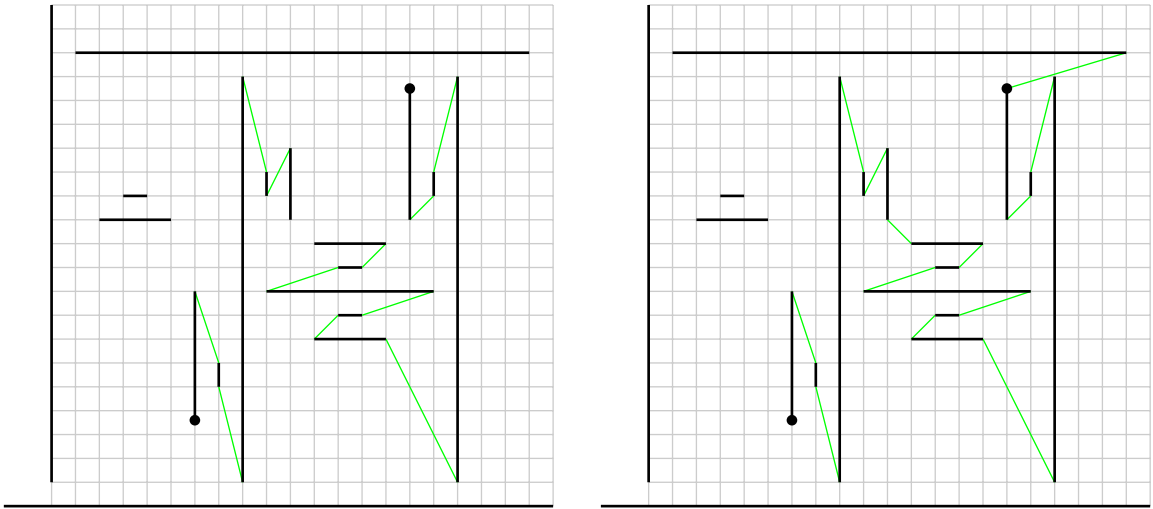}
	\caption{Linking the top horizontal segment.
	To avoid creating a loop or a dead end,
	the highest segment in the middle group of five horizontal segments must now be linked
	to the rightmost segment in the left group of five vertical segments,
	and the two sequences merge into a single sequence $5^+,1,17,1,3,3,1,7,1,3,17,1,5^+$.
	Then the endpoint $v' = (15,17+\frac13+\delta)$ must be linked to the right endpoint $(20,19)$
	of the top horizontal segment;
	linking it to the left endpoint $(1,19)$ would block
	further linking to the length-$1$ segment
	in the left group of two horizontal segments.}\label{fig:gadget4}
\end{figure}

\begin{figure}[p]
	\centering\includegraphics{gadget4.eps}\includegraphics{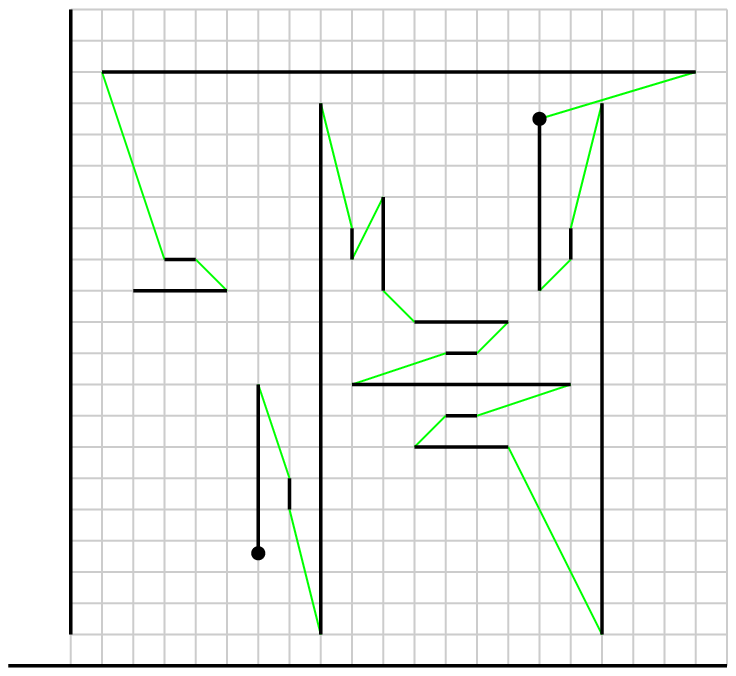}
	\caption{Linking the left group of two horizontal segments.
	As required by the length-$1$ horizontal segment in the left group,
	the sequence must extend to $5^+,1,17,1,3,3,1,7,1,3,17,1,5^+,19,1,3$.}\label{fig:gadget5}
\end{figure}

\begin{figure}[p]
	\centering\includegraphics{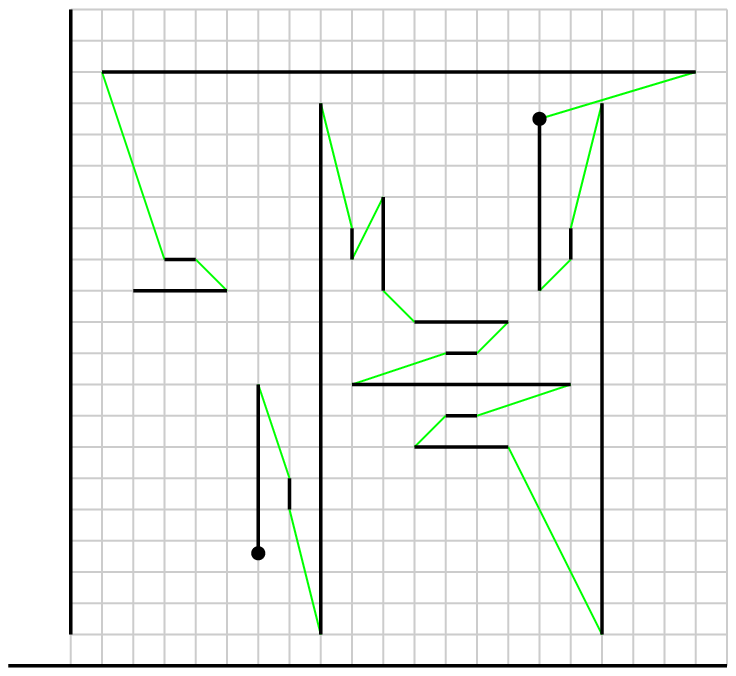}\includegraphics{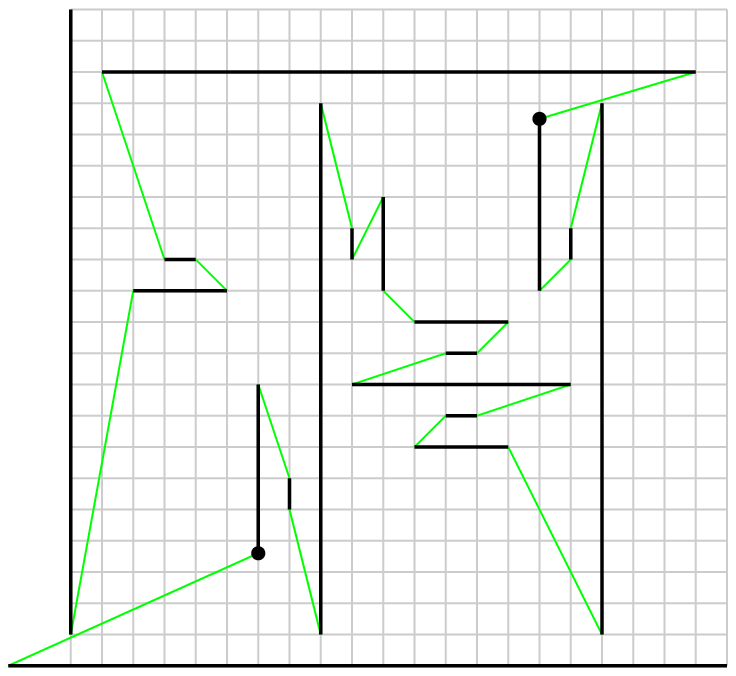}
	\caption{Linking the horizontal segment $o'a$ and the vertical segment $o''b$.
	Finally,
	the sequence extends at both ends to $o'$ and $o''$.
	In particular, the length-$3$ horizontal segment in the left group must be
	linked to $o'' = (0,1)$,
	and the length-$5^+$ vertical segment in the left group of five vertical segments
	must be linked to $o' = (-2,0)$ through its lower endpoint $u' = (6,4-\delta)$.
	}\label{fig:gadget6}
\end{figure}

\clearpage

Recall that the coordinates of endpoints in $V(\S)$
are integers of magnitude $O(n^2)$.
After the transformation,
all endpoints in $V(\S')$ except the black dots
have integer coordinates too.
We can scale all coordinates by another factor of $O(1/\delta) = O(n^2)$,
so that all endpoints in $V(\S')$ including the black dots
have integer coordinates of magnitude $O(n^4)$,
and the reduction remains strongly polynomial.
Thus \scirc\ is strongly NP-hard,
even if the input segments are disjoint and axis-parallel.

\section{Modification for \spath}

To prove that \spath\ is also NP-hard,
we use almost the same transformation from $\S$ to $\S'$ as before,
except two changes:
\begin{enumerate}\setlength\itemsep{0pt}

		\item
			Increase the initial scaling factor from $42$ to $80$,
			and correspondingly decrease the distance $\delta$
			when constructing the gadgets, to ensure that Lemma~\ref{lem:angle} still holds.

		\item
			Select an arbitrary gadget,
			and replace it by an extended gadget of the same orientation.
			For example, if the gadget is as illustrated in Figure~\ref{fig:gadget},
			then the extended gadget is as illustrated in Figure~\ref{fig:extended}.
			A closer look at the construction~\cite[Figure~8]{Ra89} shows that there is at least one
			intersection in $\S$; correspondingly, we have at least one gadget in $\S'$ ready for
			this upgrade.

			The endpoints of the segments inside the extended gadget have integer coordinates
			as indicated by the grid lines, except six endpoints illustrated as black dots:
		\begin{itemize}\setlength\itemsep{0pt}
				\item
					the lower endpoint of the vertical segment at $x = 6$ is $u' = (6, 4 - \delta)$,
				\item
					the two vertical segments at $x = 8$ and $x = 38$ have
					lower endpoints at $y = \epsilon$ and upper endpoints at $y = 18 - \epsilon$,
					where $\epsilon = \frac1{1000}$.
				\item
					the upper endpoint of the vertical segment at $x = 36$ is
					$v'' = (36,18 - 3\epsilon + \delta)$.
		\end{itemize}

\end{enumerate}

The purpose of the first change is to make room for additional segments
in the extended gadget.
Note that the endpoint $v'' = (39,18 - 3\epsilon + \delta)$ is at a distance of $\delta$ 
above the point $(36,18 - 3\epsilon)$ which is collinear with the two endpoints
$(38, 18 - \epsilon)$ and $(39, 18)$.
By an analogous argument as in Lemma~\ref{lem:angle},
with $\delta = 1/(c\cdot n^2)$ for a sufficiently large integer constant $c > 0$,
we can ensure that $u'$ and $v''$ cannot see any endpoints outside the extended gadget.
Clearly, the reduction remains strongly polynomial.

\begin{figure}[htbp]
	\centering\includegraphics{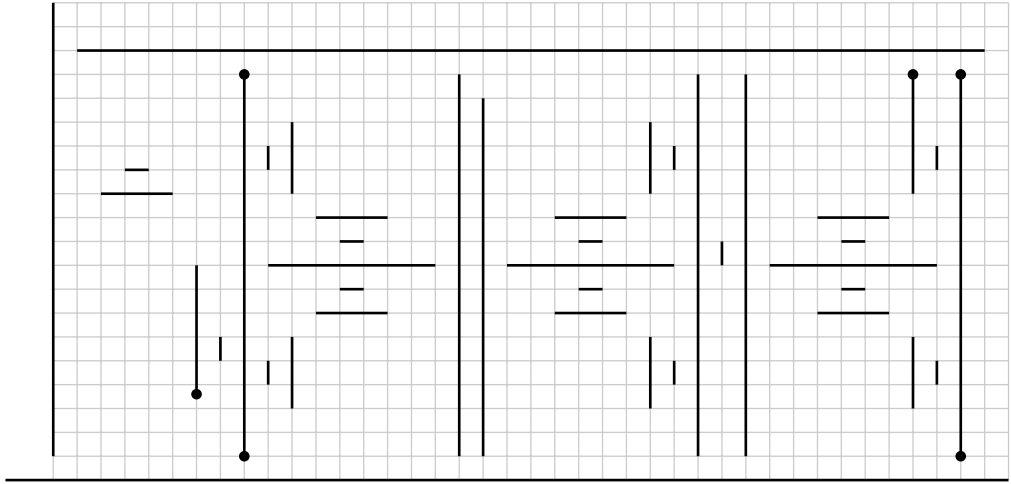}
	\caption{The extended gadget for the intersection $o$ between a horizontal segment $oa$
	and a vertical segment $ob$ illustrated on a $40\times 20$ grid
	(due to the scaling factor of $80$, there is a free space of $40\times 40$ for each gadget,
	which is sufficient for both the $40\times 20$ grid of the extended gadget
	and the $21\times 21$ grid of the ordinary gadgets).
	The common endpoint $o = (0, 0)$ is split to $o' = (-2,0)$, $o'' = (0, 1)$.
	The top horizontal segment has left endpoint $(1, 18)$ and right endpoint $(39, 18)$.
	The six endpoints illustrated as black dots are, from left to right,
	$u' = (6, 4-\delta)$, $(8, \epsilon)$, $(8, 18-\epsilon)$,
	$v'' = (36, 18 - 3\epsilon + \delta)$, $(38, \epsilon)$, $(38, 18 - \epsilon)$.
	Except the leftmost black dot $u'$,
	the other five black dots are illustrated with integer $y$-coordinates
	here and in subsequent figures for visual clarity.}\label{fig:extended}
\end{figure}

The following lemma is analogous to Lemma~\ref{lem:link}:

\begin{lemma}
	$\S$ can be linked into a simple polygon
	if and only if $\S'$ can be linked into a simple polygonal chain.
\end{lemma}

\begin{proof}
	We first prove the direct implication.
	Suppose that $\S$ can be linked into a simple polygon.
	We link the segments in $\S'$ as before, except that in the extended gadget
	we link the segments as illustrated in Figure~\ref{fig:extended4}.
	Then $\S'$ is linked into a simple polygonal chain starting and ending in the extended gadget.

	We next prove the reverse implication.
	Suppose that $\S'$ can be linked into a simple polygonal chain.
	As before, the linking of the segments in each ordinary gadget
	into an alternating path is inevitable.
	Refer to Figures~\ref{fig:extended1},
	\ref{fig:extended2},
	\ref{fig:extended3},
	\ref{fig:extended4}.
	The linking of the segments in the extended gadget into two disjoint alternating paths
	is also inevitable.
	Then, following the other visibility edges of the simple polygonal chain through $\S'$,
	which are outside and between the gadgets,
	$\S$ can be linked into a simple polygon.
\end{proof}

\begin{figure}[htbp]
	\centering\includegraphics{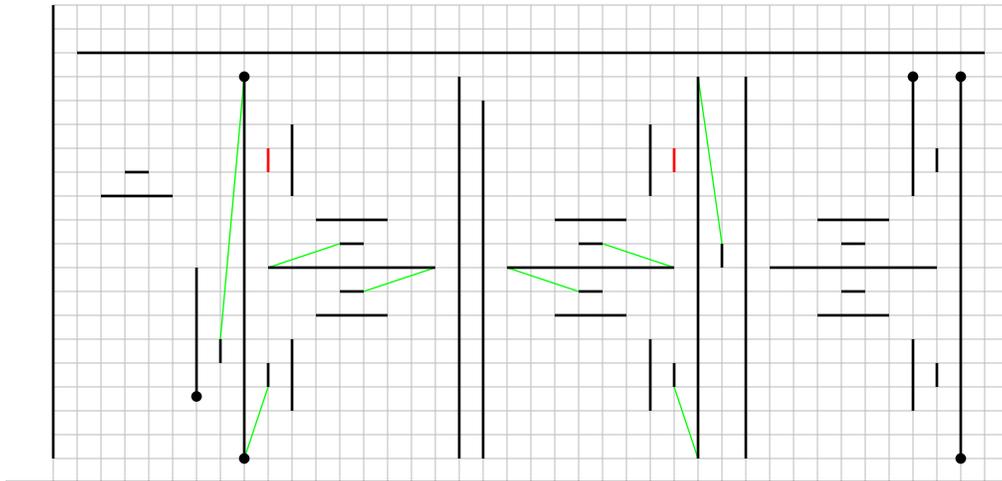}
	\caption{Consider the left group of five horizontal segments
	and the group of seven vertical segments to their left.
	There are five length-$1$ segments among these $12$ segments.
	Each of these length-$1$ segments can see either one or both of
	the length-$(18-2\epsilon)$ vertical segment and the length-$7$ horizontal segment,
	plus one other segment.
	Since the two long segments can accommodate at most four neighbors,
	at least one of the five length-$1$ segments must be linked to only one neighbor,
	and hence is either the starting or the ending segment of the polygonal chain.
	The situation is similar for the middle group of five horizontal segments
	and the group of seven vertical segments to their right.
	Among these $12$ segments, there are also five length-$1$ segments,
	and one of the five must be either the starting or the ending segment of the polygonal chain.
	For example, the polygonal chain could start and end at the two intervals illustrated in red,
	then the other eight length-$1$ segments must be linked to
	the two long vertical segments and the two long horizontal segments,
	which become unavailable for further linking.
	In particular, the length-$(18-2\epsilon)$ vertical segment on the left is now a barrier.
	}\label{fig:extended1}
\end{figure}

\begin{figure}[htbp]
	\centering\includegraphics{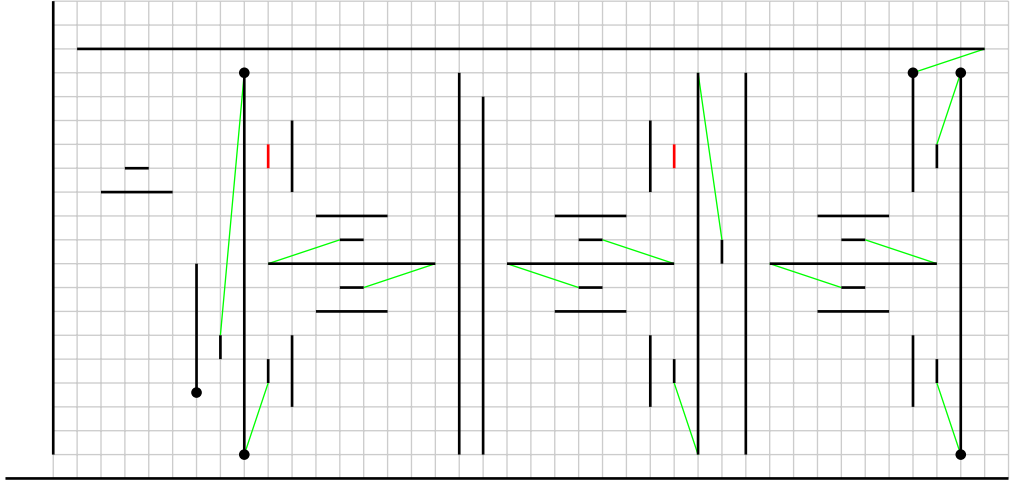}
	\caption{Since the starting and the ending segments of the polygonal chain
	have both been accounted for,
	each remaining segment must be linked to two neighbors.
	Consider the right group of five horizontal segments
	and the group of five vertical segments to their right.
	Among these $10$ segments, there are four length-$1$ segments
	which must be linked to the length-$7$ horizontal segment
	and the length-$(18-2\epsilon)$ vertical segment.
	Then this length-$(18-2\epsilon)$ vertical segment also becomes a barrier.
	Since the two length-$(18-2\epsilon)$ vertical segments at $x = 8$ and $x = 38$
	are only a distance of $\epsilon = \frac1{1000}$ away
	from the two horizontal segments at $y = 0$ and $y = 18$,
	the segments bounded by them are all isolated, except the black dot
	$v'' = (36, 18-3\epsilon + \delta)$,
	which must be linked to the right endpoint $(39,18)$ of the top horizontal segment.
	}\label{fig:extended2}
\end{figure}

\begin{figure}[htbp]
	\centering\includegraphics{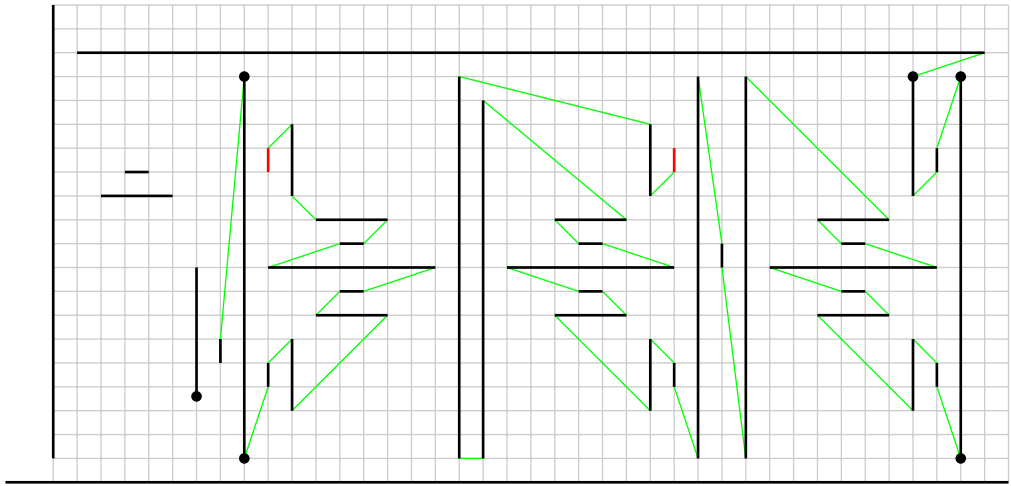}
	\caption{The isolated segments have only two ways out,
	one through the length-$(18-2\epsilon)$ vertical segment on the left,
	and the other through the endpoint $v''$ on the right.
	Thus all of them must be linked internally into two disjoint alternating paths.
	}\label{fig:extended3}
\end{figure}

\begin{figure}[htbp]
	\centering\includegraphics{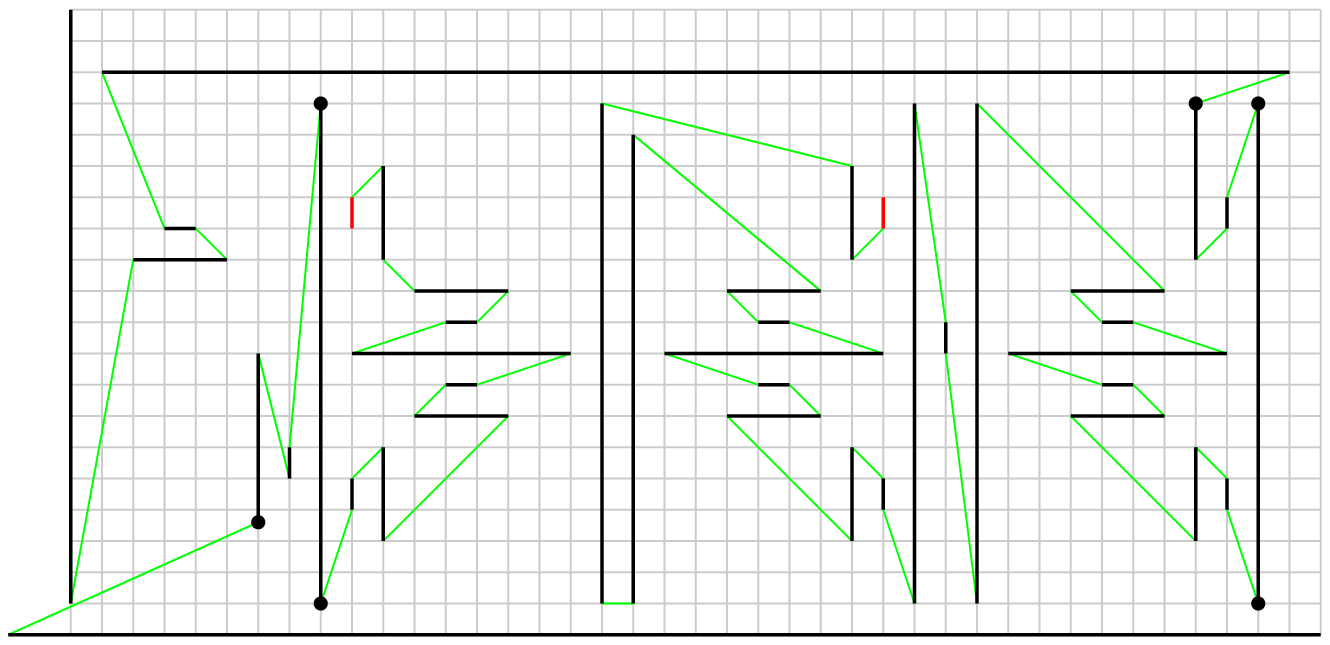}
	\caption{The two disjoint alternating paths must then extend through the remaining segments,
	and finally exit the extended gadget at $o'$ and $o''$,
	and will be joined into a single polygonal chain
	through the segments outside the extended gadget.
	}\label{fig:extended4}
\end{figure}

Thus \spath\ is also NP-hard.
This completes the proof of Theorem~\ref{thm:hard}.

\end{document}